\documentclass[11pt]{article}

\usepackage{fullpage}
\usepackage{bbm}
\usepackage{graphicx}
\usepackage{upref}
\usepackage{enumerate}
\usepackage{latexsym}
\usepackage{color,graphics}
\usepackage{comment} 
\usepackage{relsize}

\usepackage[section]{algorithm} 
\usepackage{amsmath,amsthm,amssymb,algorithmic,epsfig,graphicx, tikz}

\usepackage{amsfonts}
\usepackage{varioref}
\usepackage[ansinew]{inputenc}

\usepackage{amsmath}
\usepackage{amsfonts}
\usepackage{tikz} 

\usepackage{amssymb}
\usepackage{varioref}
\usepackage[ansinew]{inputenc}

\newtheorem{theorem}{Theorem}[section]
\newtheorem{lemma}[theorem]{Lemma}

\newtheorem{defn}[theorem]{Definition}

\newcommand{\R}{\mathbb{R}}

\def\ket#1{\mathinner{|{#1}\rangle}}
\newcommand{\braket}[2]{\langle #1|#2\rangle}

\renewcommand{\part}[2]{\frac{\partial #1}{\partial #2}}

\renewcommand {\ss}[1] { \subsection{#1} } 
\newcommand{\als}[1]{\begin{align*}#1\end{align*}}
\newcommand{\all}[2]{\begin{align}\label{#2} #1\end{align}}
\newcommand{\al}[1]{\begin{align} #1\end{align}}

\newcommand{\en}[1]{\left ( #1 \right )}

\newcommand{\rankk} {rank-$k$ } 
\newcommand{\topk} {top-$k$ }

\newcommand{\nl}{\notag \\}

\newcommand{\norm}[1]{\lVert#1\rVert}

\newcommand{\thmref}[1]{\hyperref[#1]{{Theorem~\ref*{#1}}}}
\newcommand{\lemref}[1]{\hyperref[#1]{{Lemma~\ref*{#1}}}}
\newcommand{\remref}[1]{\hyperref[#1]{{Remark~\ref*{#1}}}}
\newcommand{\corref}[1]{\hyperref[#1]{{Corollary~\ref*{#1}}}}
\newcommand{\eqnref}[1]{\hyperref[#1]{{Equation~(\ref*{#1})}}}
\newcommand{\claimref}[1]{\hyperref[#1]{{Claim~\ref*{#1}}}}
\newcommand{\remarkref}[1]{\hyperref[#1]{{Remark~\ref*{#1}}}}
\newcommand{\propref}[1]{\hyperref[#1]{{Proposition~\ref*{#1}}}}
\newcommand{\factref}[1]{\hyperref[#1]{{Fact~\ref*{#1}}}}
\newcommand{\defref}[1]{\hyperref[#1]{{Definition~\ref*{#1}}}}
\newcommand{\exampleref}[1]{\hyperref[#1]{{Example~\ref*{#1}}}}
\newcommand{\hypref}[1]{\hyperref[#1]{{Hypothesis~\ref*{#1}}}}
\newcommand{\secref}[1]{\hyperref[#1]{{Section~\ref*{#1}}}}
\newcommand{\chapref}[1]{\hyperref[#1]{{Chapter~\ref*{#1}}}}
\newcommand{\apref}[1]{\hyperref[#1]{{Appendix~\ref*{#1}}}}

\AtBeginDocument{}

\begin{document}
\title{Quantum Recommendation Systems}  

\author{ 
Iordanis Kerenidis \thanks{
CNRS, IRIF, Universit\'e Paris Diderot, Paris, France and  
Centre for Quantum Technologies, National University of Singapore, Singapore. 
Email: {\tt jkeren@liafa.univ-paris-diderot.fr}.} 
\and
Anupam Prakash \thanks{Centre for Quantum Technologies and School of Physical and Mathematical Sciences, Nanyang Technological University, Singapore.
Email: { \tt aprakash@ntu.edu.sg}.}
}

\maketitle

\begin{abstract} 
A recommendation system uses the past purchases or ratings of $n$ products by a group of $m$ users, in order to provide personalized recommendations to individual users. The information is modeled as an $m \times n$ preference matrix which is assumed to have a good \rankk approximation, for a small constant $k$. 

In this work, we present a quantum algorithm for recommendation systems that has running time $O(\text{poly}(k)\text{polylog}(mn))$. All known classical algorithms for recommendation systems that work through reconstructing an approximation of the preference matrix run in time polynomial in the matrix dimension. Our algorithm provides good recommendations by sampling efficiently from an approximation of the preference matrix, without  reconstructing the entire matrix. For this, we design an efficient quantum procedure to project a given vector onto the row space of a given matrix. This is the first algorithm for recommendation systems that runs in time polylogarithmic in the dimensions of the matrix and provides an example of a quantum machine learning algorithm for a real world application. 

\end{abstract} 

\section{Introduction}

A recommendation system uses information about past purchases or ratings of products by a group of users in order to provide personalized recommendations to individual users. More precisely, we assume there are $m$ users, for example clients of an online platform like Amazon or Netflix, each of whom have some inherent preference or utility about $n$ products, for example books, movies etc. The user preferences are modeled by an $m \times n$ matrix $P$, where the element $P_{ij}$ denotes how much the user $i$ values product $j$. If the preference matrix $P$ had been known in advance, it would have been easy to make good recommendations to the users by selecting elements of this matrix with high value. However, this matrix is not known a priori. Information about $P$ arrives in an online manner each time a user buys a product, writes a review, or fills out a survey. A recommendation system tries to utilize the already known information about all users in order to suggest products to individual users that have high utility for them and can eventually  lead to a purchase. 

There has been an extensive body of work on recommendation systems, since it is a very interesting theoretical problem and also of great importance to the industry. We cite the works of \cite{AFKMS01, PTRV98, DKR02, APPT05} who studied the problem in a combinatorial or linear algebraic fashion. There has also been a series of works in the machine learning community many of them inspired by a practical challenge by Netflix on real world data \cite{KBV09, BK07, BK11}.

We next discuss the low rank assumption on the preference matrix underlying recommendation systems and the 
way this assumption is used to perform matrix reconstruction in classical recommendation systems. We then describe the 
computational model for our quantum recommendation algorithm that is based on matrix sampling and compare 
it to classical recommendation algorithms based on matrix reconstruction. We provide a high level overview of our algorithm 
in section \ref{oneone} and then, we compare it with previous work on quantum machine learning in section \ref{onetwo}. 

\paragraph*{The low-rank assumption.}
The underlying assumption in recommendation systems is that one can infer information about a specific user from the information about all other users because, in some sense, the majority of users belong to some well-defined ``types". In other words, most people's likes are not unique but fall into one of a small number of categories. Hence, we can aggregate the information of ``similar" users to predict which products have high utility for an individual user. 

More formally, the assumption in recommendation systems is that the preference matrix $P$ can be well approximated (according to some distance measure) by a low-rank matrix.  There are different reasons why this assumption is indeed justified. First, from a philosophical and cognitive science perspective, it is believed that there are few inherent reasons why people buy or not a product: the price, the quality, the popularity, the brand recognition, etc. (see for example \cite{AFKMS01, R79}). Each user can be thought of as weighing these small number of properties differently but still, his preference for a product can be computed by checking how the product scores in these properties. Such a model produces a matrix $P$ which has a good \rankk approximation, for a small $k$, which can be thought of as a constant independent of the number of users $m$ or the number of products $n$. 
Moreover, a number of theoretical models of users have been proposed in the literature which give rise to a matrix with good low-rank approximation. For example, if one assumes that the users belong to a small number of ``types", where a type can be thought of as an archetypical user, and then each individual user belonging to this type is some noisy version of this archetypical user, then the matrix has a good low-rank approximation \cite{DKR02, PTRV98}. In addition, preference matrices that come from real data have been found to have rank asymptotically much smaller than the size of the matrix.

For these reasons, the assumption that the matrix $P$ has a good low-rank approximation has been widely used in the literature. In fact, if we examine this assumption more carefully, we find that in order to justify that the recommendation system provides high-value recommendations, we assume that users ``belong" to a small number of user types and also that they agree with these types on the high-value elements. For contradiction, imagine that there are $k$ types of users, where each type has very few high-value elements and many small value elements. Then, the users who belong to each type can agree on all the small value elements and have completely different high-value elements. In other words, even though the matrix is low-rank, the recommendations would be of no quality. Hence, the assumption that has been implicitly made, either by philosophical reasons or by modeling the users, is that there are $k$ types of users and the users of each type ``agree" on the high-value elements.

\vspace{-0.3cm}
\paragraph*{Recommendations by Matrix Reconstruction.}
One of the most powerful and common ways to provide competitive recommendation systems is through a procedure called matrix reconstruction. In this framework, we assume that there exists a hidden matrix $A$, in our case the preference matrix, 
which can be well approximated by a low-rank matrix. The reconstruction algorithm gets as input a number of samples from $A$, in our case the previous data about the users' preferences, and outputs a \rankk matrix with the guarantee that it is ``close" to $A$ according to some measure (for example, the $2$- or the Frobenius norm). For example, the reconstruction algorithm can perform a Singular Value Decomposition on the subsample matrix $\widehat{A}$, where $\widehat{A}$ agrees with $A$ on known samples and is $0$ on the remaining entries, and output the projection of $\widehat{A}$ onto the space spanned by its top-$k$ singular vectors. The ``closeness" property guarantees that the recommendation system will select an element that with high probability corresponds to a high-value element in the matrix $A$ and hence it is a good recommendation (\cite{DKR02, AFKMS01}). Another commonly used algorithm for matrix reconstruction is a variant of 
alternating minimization, this has been successful in practice \cite{KBV09} and has been recently analyzed theoretically \cite{JNS13}. Note that all known algorithms for matrix reconstruction require time polynomial in the matrix dimensions. 

An important remark is that matrix reconstruction is a harder task than recommendation systems, in the sense that a good recommendation system only needs to output a high value element of the matrix and not the entire matrix \cite{RU12, GT05}. Nevertheless, classical algorithms perform a reconstruction of the entire matrix as the resources required for finding high value elements are the same as the resources
needed for full reconstruction. 
\vspace{-0.3cm}
\paragraph*{Computational resources and performance.}
In order to make a precise comparison between classical recommendation systems and our proposed system, we discuss more explicitly the computational resources in recommendation systems. 
We are interested in systems that arise in the real world, for example on Amazon or Netflix, where the number of users can be about 100 million and the products around one million. 
For such large systems, storing the entire preference matrix or doing heavy computations every time a user comes into the system is prohibitive. 

The memory model for an online recommendation system is the following. A data structure is maintained 
that contains the information that arrives into the system in the form of elements $P_{ij}$ of the preference matrix. We require that the time needed to write the tuple $(i,j,P_{ij})$ into the memory data structure and to read it out is polylogarithmic in the matrix dimensions. In addition, we require that the total memory required is linear (up to polylogarithmic terms) in the number of entries of the preference matrix that have arrived into the system. For example, one could store the elements $(i,j,P_{ij})$ in an ordered list. 

Most classical recommendation systems use a two stage approach. The first stage involves preprocessing the data stored in memory. For example, a matrix reconstruction algorithm can be performed during the preprocessing stage to produce and store a low-rank approximation of the preference matrix. This computation takes time polynomial in the matrix dimensions, $\mbox{poly}(mn)$, and the stored output is the \topk row singular vectors that need space $O(nk)$. 
The second stage is an online computation that is performed when a user comes into the system. For example, one can project the row of the subsample matrix that corresponds to this user onto the already stored \topk row singular vectors of the matrix and output a high value element in time $O(nk)$. 

The goal is to minimize the time needed to provide an online recommendation while at the same time keeping the time and the extra memory needed for the preprocessing reasonable. In general, the preprocessing time is polynomial in the dimensions of the preference matrix, i.e. $\mbox{poly}(mn)$, the extra memory is $O(nk)$, while the time for the online recommendation is $O(nk)$. Note that in real world applications, it is prohibitive to have a system where the preprocessing uses memory $O(mn)$, even though with such large memory the online recommendation problem becomes trivial as  all the answers can be pre-computed and stored.

A recommendation system performs well, when with high probability and for most users it outputs a recommendation that is good for a user. 
The performance of our recommendation system is similar to previous classical recommendation systems based on matrix reconstruction and depends on how good the low-rank approximation of the matrix is. Our algorithm works for any matrix, but as in the classical case, it guarantees good recommendations only when the matrix has a good low-rank approximation.  


\subsection{Our results} \label{oneone} 

In this section, we provide a high level overview of our quantum recommendation algorithm 
which requires time polylogarithmic in the matrix dimensions and polynomial only in the rank of the matrix, which as we have argued is assumed to be much smaller than the dimension of the matrix. This is the first algorithm for recommendation systems with complexity polylogarithmic in the matrix dimensions. 
\vspace{-0.3cm}
\paragraph*{Our model.}

First, we describe a simple and general model for online recommendation systems.  
We start with a hidden preference matrix $T$, where the element $T_{ij}$ takes values 0 or 1 and indicates whether 
product $j$ is "good" for user $i$. Such boolean matrices arise naturally in a "thumbs up / thumbs down" system, where users can declare whether they like or not a certain product. We can also easily construct such matrices from non-boolean preference matrices. For each user, we split the products into two categories, the ``good" and the ``bad" recommendations, based on the matrix entries. This categorization can be done in different ways and we do not have to impose any constraints. For example, good recommendations can be every product with value higher than a threshold or the 100 products with the highest values etc.

Our assumption is that the matrix $T$ has a good low-rank approximation. The reasons that justify this assumption are the ones used already in the literature. As before, we believe that there is a small number of user types, and within each type the users ``agree" on the high-value elements. This modelling of the users gives rise to a matrix $T$ with a good low-rank approximation.
Once we have defined the matrix $T$, then any algorithm that reconstructs a matrix close to $T$, will provide a good recommendation, since $T$ is the indicator matrix of good recommendations. 
\vspace{-0.3cm}
\paragraph*{Recommendations by Matrix Sampling.}
The low-rank approximation of the matrix $T$ can be computed as follows: first, define the matrix $\widehat{T}$, where with some probability each element of $\widehat{T}$ is equal to the corresponding element in $T$ normalized and otherwise it is zero. This matrix, that we call a subsample matrix, corresponds to the information the recommendation system has already gathered about the matrix $T$. Then, by performing a Singular Value Decomposition and computing the projection of this matrix to its \topk row singular vectors, we compute a matrix $\widehat{T}_k$ which can be proven to be close to the matrix $T$, as long as $T$ had a good \rankk approximation.

As we remarked, in principle, we do not need to explicitly compute the entire matrix $\widehat{T}_k$. 
It is sufficient to be able to sample from the matrix $\widehat{T}_k$ which is close to $T$. Since $T$ is a 0-1 matrix, sampling from $\widehat{T}_k$ means finding with high probability a 1-element in $T$. By the fact that $T$ indicates the good recommendations, our algorithm will output a good recommendation with high probability. 
Hence, we reduce the question of providing good recommendations to being able to sample from the matrix $\widehat{T}_k$. In fact, since we want to be able to recommend products to any specific user $i$, we need to be able, given an index $i$, to sample from the $i$-th row of the matrix $\widehat{T}_k$, denoted by $ ( \widehat{T}_k )_i$,
i.e. output an element $( \widehat{T}_k )_{ij}$ with probability $| ( \widehat{T}_k )_{ij}|^2 / \norm{( \widehat{T}_k )_i}^2$. Note that the row $( \widehat{T}_k)_i$ is the projection of the row $\widehat{T}_i$ onto the \topk row singular vectors of $\widehat{T}$.

\vspace{-0.3cm}
\paragraph*{An efficient quantum algorithm for Matrix Sampling.}
Here is where quantum computing becomes useful: we design a quantum procedure that samples from the row $ ( \widehat{T}_k )_i$ in time $\mbox{polylog}(m n)$. Note that the quantum algorithm does not output the row $ ( \widehat{T}_k )_i$, which by itself would take time linear in the dimension $n$, but only samples from this row. But this is exactly what is needed for recommendation systems: Sample a high-value element of the row, rather than explicitly output the entire row. More precisely, we describe an efficient quantum procedure that takes as input a vector, a matrix, and a threshold parameter
and generates the quantum state corresponding to the projection of the vector onto the space spanned by the row singular vectors of the matrix whose corresponding singular value is greater than the threshold.
From the outcome of this procedure it is clear how to sample a product by just measuring the quantum state in the computational basis. 

 \subsection{Comparisons with related work.} \label{onetwo} 
The development of quantum algorithms for linear algebra was initiated by the breakthrough algorithm of Harrow, Hassidim, Lloyd \cite{HHL09}. The HHL algorithm takes as input a sparse (the number of non zero entries in each row of the matrix is polylogarithmic) and well-conditioned system of linear equations and in time polylogarithmic in the dimension of the system outputs a quantum state which corresponds to the classical solution of the system. Note that this algorithm does not explicitly output the classical solution, nevertheless, the quantum state enables one to sample from the solution vector. This is a very powerful algorithm and 
has been very influential in recent times, where several works \cite{LMR13, LMR13a, LMR13b} obtained quantum algorithms for machine learning problems based on similar assumptions. However, when looking at these applications, one needs to be extremely careful about two things: first, the assumptions that one needs to make on the input in order to achieve efficient running time, since, for example, the running time of the HHL algorithm is polylogarithmic only when the matrix is well conditioned (i.e. the minimum singular value is at least inverse polynomially big) and sparse; and second, whether the quantum algorithm solves the original classical problem or a weaker variant to account for the fact that the classical solution is not given explicitly but is encoded in a quantum state \cite{A15,LMR13b}. In addition, we mention a recent but orthogonal proposal to use techniques inspired by the structure of quantum theory for classical recommender systems \cite{S16}. 

Let us be more explicit about our algorithm's assumptions. We assume the data is stored in a classical data structure which enables the quantum algorithm to efficiently create superpositions of rows of the subsample matrix. The HHL algorithm also needs to be able to efficiently construct quantum states from classical vectors given as inputs. In the Appendix, 
we describe a classical data structure for storing the matrix $\widehat{T}$. The data structure maintains some extra information about the matrix entries, so that, the total memory needed is linear (up to polylogarithmic terms) in the number of entries in the subsample matrix, the data entry time remains polylogarithmic in the matrix dimensions, and an algorithm with quantum access to the data structure can create the necessary superpositions in polylogarithmic time. Note also, that even in the case the data has been stored as a normal array or list, we can preprocess it in linear time to construct our needed data structure. Thus, our quantum algorithm works under the same memory model as any other quantum query algorithm (e.g. Grover's algorithm): it assumes that there exists a classical data structure to which we can make quantum queries.
Overall, our system retains the necessary properties for the data entry and retrieval stage. Moreover, the classical complexity of matrix reconstruction does not change given the new data structure. 

Importantly, in our system, we do not perform any preprocessing nor do we need any extra memory. Our recommendation algorithm just performs an online computation that requires time $\mbox{poly}(k)\mbox{polylog}(mn)$. This can be viewed as exponentially smaller than the classical time if the rank $k$ is a small 
constant and the matrix dimensions are of the same order. 
Unlike the HHL algorithm, our running time does not depend on the sparsity of the input matrix nor on its condition number, i.e. its smallest singular value. In other words, we do not make any additional assumptions about the classical data
beyond the low rank approximation assumptions made by classical recommendation systems.

It is also crucial to note that we have not changed what one needs to output, as was the case for the HHL algorithm and its applications, where instead of explicitly outputting a classical solution, they construct a quantum state that corresponds to this solution. We have instead described a real world application, where the ability to sample from the solution is precisely what is needed. 

The rest of this paper is organized as follows. We introduce some preliminaries in section 2. In sections 3 and 4
we show that sampling from an approximate reconstruction of the matrix $T$ suffices to provide good recommendations and that if the sub-samples $\widehat{T}$ are uniformly distributed then projecting onto 
the top $k$ singular vectors of $\widehat{T}$ is an approximate reconstruction for $T$. In section 5 we describe 
an efficient quantum algorithm for projecting a vector onto the space of singular vectors of $\widehat{T}$ whose corresponding singular values are greater than a threshold. In section 6 we combine these components to obtain a quantum recommendation algorithm and analyze its performance and running time. 


\section{Preliminaries}

\subsection{Linear algebra} 
 The set $\{ 1, 2, \cdots, n\}$ is denoted by $[n]$, the standard basis vectors in $\R^{n}$ are denoted by $e_{i}, i \in [n]$.
For any matrix $A\in \R^{m\times n}$, the Frobenius norm is defined as $\norm{A}_{F}^{2}= \sum_{ij} A_{ij}^{2} = \sum_{i} \sigma_{i}^{2}$, where $\sigma_{i}$ are the singular values. We also say that we sample from the matrix $A$ when we pick an element $(i,j)$ with probability $|A_{ij}|^2/\norm{A}_F^2$, and write $(i,j) \sim A$. For a vector $x \in \R^n$ we denote the norm $\norm{x} ^2 = \sum_{i} x_i^2$.

The matrix $A$ is unitary if $AA^{*}=A^{*}A=I$, the eigenvalues of a unitary matrix 
have unit norm. A matrix $P \in \R^{n \times n}$ is a projector if $P^{2}=P$. If $A$ is a matrix with orthonormal columns, 
then $AA^{t}$ is the projector onto the column space of $A$. 

\textbf{Singular value decomposition:} 
The singular value decomposition of $A\in \R^{m\times n}$ is a decomposition of the form $A=U\Sigma V^{t}$ where 
$U \in \R^{m\times m}, V\in \R^{n\times n}$ are unitary and $\Sigma \in \R{m \times n}$ is a diagonal matrix with positive entries.  
The $SVD$ can be written as  $A = \sum_{i \in [r]} \sigma_{i} u_{i} v_{i}^{t}$ where $r$ is the rank of $A$. 
The column and the row singular vectors $u_{i}$ and $v_{i}$ are the columns of $U$ and $V$ respectively. The Moore Penrose pseudo-inverse is 
defined as $A^{+} = V \Sigma^{+} U^{t}$, where $A^{+} = \sum_{i\in [r]} \frac{1}{\sigma_{i}} v_{i} u_{i}^{t}$. 
It follows that $AA^{+}$ is the projection onto the column space $Col(A)$ while $A^{+}A$ is the projection onto the row space $Row(A)$. The truncation of $A$ to the space of the singular vectors that correspond to the $k$ largest singular values is denoted by $A_{k}$, that is $A_{k}=  \sum_{i \in [k]} \sigma_{i} u_{i} v_{i}^{t}$. We denote by $A_{ \geq \sigma}$ the projection of the matrix $A$ onto the space spanned by the singular vectors whose corresponding singular value is bigger than $\sigma$, that is $A_{\geq \sigma}=  \sum_{i : \sigma_i \geq \sigma} \sigma_{i} u_{i} v_{i}^{t}$.

\subsection{Quantum information}
We use the standard bra-ket notation to denote quantum states. 
We use the following encoding for representing $n$ dimensional vectors by quantum states, 
\begin{defn} \label{vstate} 
The vector state $\ket{x}$ for $x \in \R^{n}$ is defined as $\frac{1}{\norm{x} } \sum_{i\in [n]} x_{i} \ket{i}$. 
\end{defn} 
\noindent
In case $x \in \R^{mn}$, we can either see it as a vector in this space or as a matrix with dimensions $m \times n$ and then we can equivalently write
$\frac{1}{\norm{x} } \sum_{i\in [m], j \in [n]} x_{ij} \ket{i,j}$.

A quantum measurement $(POVM)$ is a collection of positive operators 
$M_{a} \succeq 0$ such that $\sum_{a} M_{a} =I_{n}$. The probability of obtaining outcome $a$ 
when state $\ket{\phi}$ is measured is $Tr(\braket{\phi} {M_{a}\phi})$. 
If $\ket{x}$ is measured in the standard basis, then outcome $i$ is observed with probability 
$x_{i}^{2}/\norm{x}^{2}$.

We also use a well-known quantum algorithm called phase estimation. The phase estimation algorithm estimates the eigenvalues of a 
unitary operator $U$ with additive error $\epsilon$ in time $O(  T(U) \log n /\epsilon)$ if $T(U)$ is the time required to implement the 
unitary $U$. 
\begin{theorem} \label{pest} 
{\em Phase estimation \cite{K95}}: Let $U$ be a unitary operator, with eigenvectors $\ket{v_j}$ and eigenvalues $e^{ \iota \theta_{j}}$ for $\theta_{j} \in [-\pi, \pi]$, i.e. we have $U\ket{v_{j}} = e^{ \iota \theta_{j}} \ket{v_{j}}$ for $j \in [n]$. For a precision parameter $\epsilon >0$, there exists a quantum algorithm that runs in time $O(  T(U) \log n /\epsilon)$ and with probability $1-1/\emph{poly}(n)$ 
maps a state $\ket{\phi} = \sum_{j \in [n]} \alpha_{j} \ket{v_{j}}$ 
to the state $\sum_{j \in [n]} \alpha_{j} \ket{v_{j}}\ket{ \overline{\theta_{j}} }$ such that $\overline{\theta_{j}} \in \theta_{j} \pm \epsilon$ for all $j\in [n]$. 
\end{theorem} 
Note that we use $\iota$ to denote the imaginary unit $i$ to avoid confusion with summation indices. 
The analysis of phase estimation shows that the algorithm outputs a discrete valued estimate for each eigenvalue
that is within additive error $\epsilon$ with probability at least $0.8$, the probability is boosted to $1- 1/\text{poly}(n)$ by repeating $O(\log n)$ times and choosing the most frequent estimate.


\section{A model for recommendation systems}

\ss{The preference matrix}

We define a simple and general model for recommendation systems. 
We define a {\em preference matrix } $T$ of size $m \times n$, where every row corresponds to a user, every column to a product, and the element $T_{ij}$ is 0 or 1 and denotes whether product $j$ is a good recommendation for user $i$ or not. 

\begin{defn}
A product $j$ is a good recommendation for user $i$ iff $T_{ij} =1$, otherwise it is bad. We also write it as the pair $(i,j)$ is a good or bad recommendation.  
\end{defn}

Such matrices arise in systems where the information the users enter is binary, for example in a "thumbs up / thumbs down" system. We can also construct such matrices from more general preference matrices where the users use a star system to grade the products. One can imagine, for example, that the good recommendations could be the products for which the user has a preference higher than a threshold, or the hundred products with highest preference etc.

\subsection{Sampling an approximation of the preference matrix}

%
%


Note that sampling from the preference matrix $T$ would always yield a good recommendation, since the products that correspond to bad recommendations have probability 0. This remains true even when we want to sample from a specific row of the matrix in order to provide a recommendation to a specific user. 
Our goal now is to show that sampling from a matrix that is close to the matrix $T$ under the Frobenius norm yields good recommendations with high probability for most users.

\begin{lemma}\label{matrix}
Let $\widetilde{T}$ be an approximation of the matrix $T$ such that
$ \norm{ T - \widetilde{T} }_{F} \leq \epsilon \norm{ T }_F$. 
Then, the probability a sample according to $\widetilde{T}$ is a bad recommendation is
\[
\Pr_{(i,j) \sim \widetilde{T}} [(i,j) \mbox{ bad} ] \leq \left( \frac{\epsilon}{1-\epsilon} \right)^2
\]
\end{lemma}

\begin{proof}
By the theorem's assumption and triangle inequality, we have
\[
(1+\epsilon)\norm{T}_F \geq \norm{\widetilde{T}}_F \geq (1-\epsilon)\norm{T}_F.
\]
We can rewrite the approximation guarantee as
\al{ 
\epsilon^2  \norm{ T }^2_F \geq   \norm{ T - \widetilde{T} }_{F}^2 = \sum_{(i,j):good} (1 - \widetilde{T}_{ij})^{2} + \sum_{(i,j):bad} \tilde{T}_{ij}^2 \geq \sum_{(i,j):bad} \widetilde{T}_{ij}^2 
} 
The probability that sampling from $\widetilde{T}$ provides a bad recommendation is 
\al{ 
\Pr [(i,j) \mbox{ bad}]= \frac{ \sum_{(i,j):bad} \widetilde{T}_{ij}^2}{\norm{ \widetilde{T}} ^2_F} \leq  \frac{ \sum_{(i,j):bad} \widetilde{T}_{ij}^2}{ (1-\epsilon)^2 \norm{ T} ^2_F}  \leq \left( \frac{\epsilon}{1-\epsilon} \right)^2. 
} 
\end{proof}

The above can be rewritten as follows denoting the $i$-th row of $T$ by $T_{i}$, 
\begin{equation}\label{lemma}
\Pr[(i,j) \mbox{ bad}] = \frac{ \sum_{(i,j):bad} \widetilde{T}_{ij}^2}{\norm{ \widetilde{T}} ^2_F} = \sum_{i \in [m]} \frac{ \norm{\widetilde{T}_{i}}^{2}_{F} }{ \norm{ \widetilde{T} } ^2_F} \cdot \frac{\sum_{j:(i,j) bad} \widetilde{T}_{ij}^2}{  \norm{\widetilde{T}_{i}}^{2}_{F}  } \leq
\left( \frac{\epsilon}{1-\epsilon} \right)^2.
\end{equation}

We can see that the above lemma provides the guarantee that the probability of a bad recommendation for an average user is small, where the average is weighted according to the weight of each row. In other words, if we care more about users that have many products they like and less for users that like almost nothing, then the sampling already guarantees good performance.

While this might be sufficient in some scenarios, it would be nice to also have a guarantee that the recommendation is good for most users, where now every user has the same importance. Note that only with the closeness guarantee in the Frobenius norm, this may not be true, since imagine the case where almost all rows of the matrix $T$ have extremely few 1s and a few rows have almost all 1s. In this case, it might be that the approximation matrix is close to the preference matrix according to the Frobenius norm, nevertheless the recommendation system provides good recommendations only for the very heavy users and bad ones for almost everyone else.

Hence, if we would like to show that the recommendation system provides good recommendations for most users, then we need to assume that most users are ``typical", meaning that the number of products that are good recommendations for them is close to the average. We cannot expect to provide good recommendations for example to users that like almost nothing. One way to enforce this property is, for example, to define good recommendations for each user as the 100 top products, irrespective of how high their utilities are or whether there are even more good products for some users. In what follows we prove our results in most generality, where we introduce parameters for how many users are typical and how far from the average the number of good recommendations of a typical user can be.   

\begin{theorem}\label{rec}
Let $T$ be an $m \times n$ matrix. Let $S$ be a subset of rows of size $|S| \geq (1-\zeta)m$ (for $\zeta>0$) such that for all $i \in S$, 
\begin{equation}\label{assumption}
\frac{1}{1+\gamma}\frac{\norm{ T} _F^2}{m} \leq \norm { T_{i} }^{2}   \leq (1+\gamma) \frac{\norm{ T} _F^2}{m}
\end{equation}
for some $\gamma>0$. 
Let $\widetilde{T}$ be an approximation of the matrix $T$ such that $ \norm{ T - \widetilde{T} }_{F} \leq \epsilon \norm{  T }_F$. 
Then, there exists a subset $S' \subseteq S$ of size at least $(1-\delta-\zeta)m$ (for $\delta>0$), such that on average over the users in $S'$, the probability that a sample from the row $\widetilde{T}_{i}$ is a bad recommendation is 
\[ \Pr_{i \sim \mathcal{U}_{S'}, j \sim \widetilde{T}_i}[(i,j) \emph{ bad}] \leq \frac{\left( \frac{\epsilon(1+\epsilon)}{1-\epsilon} \right)^2}{ \left( 1/\sqrt{1+\gamma} - \epsilon/\sqrt{\delta}\right)^2  (1-\delta - \zeta)}.
\]
\end{theorem} 

\begin{proof}
We first use the guarantee that the matrices $T$ and $\widetilde{T}$ are close in the Frobenius norm to conclude that there exist at least $(1-\delta)m$ users for which 
\begin{equation} \label{frob}
 \norm{ T_{i}-\widetilde{T}_{i} } ^2 \leq \frac{\epsilon^2 \norm{ T }^2_F}{\delta m}.
\end{equation}
If not, summing the error of the strictly more than $\delta m$ users for which equation \ref{frob} is false we get the following contradiction, 
\[
 \norm{ T - \widetilde{T} } _{F}^2  > \delta m \frac{\epsilon^2 \norm{ T }^2_F}{\delta m} > \epsilon^2 \norm{ T }^2_F.
\]
Then, at least $(1-\delta-\zeta)m$ users both satisfy equation \ref{frob} and belong to the set $S$. Denote this set by $S'$.  
Using equations \eqref{assumption} and \eqref{frob} and the triangle inequality $\norm{ \widetilde{T}_{i}}   \geq \norm{T_{i}} - \norm{ T_{i}- \widetilde{T}_{i}}$, we have that for all users in $S'$
\begin{equation}\label{normS} 
\norm{ \widetilde{T}_{i}}^{2}_{F}  \geq  \frac{\norm{T}_{F}^{2}}{m} \en{ \frac{1}{\sqrt{1+\gamma}} - \frac{\epsilon}{\sqrt{\delta}} }^{2} \geq \frac{\norm{\widetilde{T}}_{F}^{2}}{(1+\epsilon)^2 m} \en{ \frac{1}{\sqrt{1+\gamma}} - \frac{\epsilon}{\sqrt{\delta}} }^{2}.
\end{equation}
We now use equations \eqref{lemma} and \eqref{normS} and have 
\begin{equation}\label{final}
\left( \frac{\epsilon}{1-\epsilon} \right)^2  \geq \sum_{i \in [m]} \frac{ \norm{ \widetilde{T}_{i}}^{2}_F }{\norm{ \widetilde{T}} ^2_F} \cdot \frac{\sum_{j:(i,j) bad} \widetilde{T}_{ij}^2}{  \norm{ \widetilde{T}_{i}}^{2}_F}\geq 
\frac{\left( 1/\sqrt{1+\gamma} - \epsilon/\sqrt{\delta}\right)^2}{(1+\epsilon)^2 m} \sum_{i \in {S'}} \frac{\sum_{j:(i,j) bad} \widetilde{T}_{ij}^2}{ \norm{ \widetilde{T}_{i}}^{2}_F}. 
\end{equation}
We are now ready to conclude that, 
\[
\Pr_{i \sim \mathcal{U}_{S'}, j \sim \widetilde{T}_i}[(i,j) \mbox{ bad}] = \frac{1}{|S'|}\sum_{i \in S'} \frac{\sum_{j:(i,j) bad} \widetilde{T}_{ij}^2}{  \norm{ \widetilde{T}_{i}}^{2}_F} 
\leq \frac{\left( \frac{\epsilon(1+\epsilon)}{1-\epsilon} \right)^2}{ \left( 1/\sqrt{1+\gamma} - \epsilon/\sqrt{\delta}\right)^2 (1-\delta-\zeta)}.
\]
\end{proof}

We note that by taking reasonable values for the parameters, the error does not increase much from the original error. For example, if we assume that $90\%$ of the users have preferences between $1/1.1$ and $1.1$ times the average, then the error over the typical users has increased by at most a factor of $1.5$. 
Note also that we can easily make the quality of the recommendation system even better if we are willing to recommend a small number of products, instead of just one, and are satisfied if at least one of them is a good recommendation. This is in fact what happens in practical systems.

\section{Matrix Sampling} 

 We showed in the previous section that providing good recommendations reduces to being able to sample from a matrix $\widetilde{T}$ which is a good approximation to the recommendation matrix $T$ in the Frobenius norm. We will now define the approximation matrix $\widetilde{T}$, by extending known matrix reconstruction techniques. The reconstruction algorithms provides good guarantees under the assumption that the recommendation matrix $T$ has a good $k$-rank approximation for a small $k$, i.e. $\norm{T-T_k}_F \leq \epsilon \norm{T}_F$ (for some small constant $\epsilon \geq 0$).

Let us now briefly describe the matrix reconstruction algorithms. In general, the input to the reconstruction algorithm is a subsample of some matrix $A$. There are quite a few different ways of subsampling a matrix, for example, sampling each element of the matrix with some probability or sampling rows and/or columns of the matrix according to some distribution. We present here in more detail the first case as is described in the work of Achlioptas and McSherry \cite{AM01}. Each element of the matrix $A$ that has size $m \times n$ is sampled with probability $p$ and rescaled so as to obtain the random matrix $\widehat{A}$ where each element is equal to $\widehat{A}_{ij}=A_{ij}/p$ with probability $p$ and 0 otherwise. Note that $E[\widehat{A}]=A$ and that it is assumed that $k$ and $\norm{A}_F$ are known.

The reconstruction algorithm computes the projection of the input matrix $\widehat{A}$ onto its $k$-top singular vectors; we denote the projection 
by $\widehat{A}_k$. The analysis of the algorithm shows that the approximation error $\norm{A - \widehat{A}_{k}}$ is not much bigger than $\norm{A - A_{k}}$. Projecting onto the top $k$ singular vectors of the subsampled matrix $\widehat{A}$ thus suffices to reconstruct a matrix approximating $A$. 

The intuition for the analysis is that $\widehat{A}$ is a matrix whose entries are independent random variables, thus with high probability the top $k$ spectrum of $\widehat{A}$ will be close to the one of its expectation matrix $E[\widehat{A}]=A$. This intuition was proven in \cite{AM01}.

\begin{theorem}\label{AM}
\cite{AM01} Let $A\in \R^{m\times n}$ be a matrix and $b=\max_{ij} A_{ij}$.
Define the matrix $\widehat{A}$ to be a random matrix obtained by subsampling with probability 
$p = 16 n b^2 / (\eta \norm{A}_F)^2$ (for $\eta>0$) and rescaling, that is $\widehat{A}_{ij} = A_{ij}/p$ with probability $p$ and $0$ otherwise. With probability at least 
$1- exp(-19 (\log n)^{4})$ we have for any $k$ 
\al{ \label{Achl}
\norm{A-\widehat{A}_{k}}_{F} &\leq   \norm{A-A_k}_F + 3\sqrt{ \eta} k^{1/4} \norm{A}_F.
}  
\end{theorem} 

Here, we will need to extend this result in order to be able to use it together with our quantum procedure.
First, we will consider the matrix which is not the projection on the $k$-top singular vectors, but the projection on the singular vectors whose corresponding singular values are larger than a threshold. For any matrix $A$ and any $\sigma \geq 0$, we denote by $A_{ \geq \sigma}$ the projection of the matrix $A$ onto the space spanned by the singular vectors whose corresponding singular value is bigger than $\sigma$. 
Intuitively, since the spectrum of the matrix is highly concentrated on the top $k$ singular vectors, then the corresponding singular values should be of order $O(\frac{\norm{A}_F}{\sqrt{k}})$. 

Note that we do not use anything about how the matrix $\widehat{A}$ was generated, only that it satisfies equation \ref{Achl}. Hence our results hold for other matrix reconstruction algorithms as well, as long as we have a similar guarantee in the Frobenius norm.

\begin{theorem}\label{sigma}
Let $A\in \R^{m\times n}$ be a matrix such that $\max_{ij} A_{ij}=1$. 
Define the matrix $\widehat{A}$ to be a random matrix obtained by subsampling with probability 
$p = 16 n/ \eta^{2} (\norm{A}_F)^2$ (for $\eta>0$) and rescaling, that is $\widehat{A}_{ij} = A_{ij}/p$ with probability $p$ and $0$ otherwise. Let $\mu>0$ a threshold parameter and denote $\sigma=\sqrt{ \frac{\mu}{k}}||\widehat{A}||_F$. With probability at least 
$1- exp(-19 (\log n)^{4})$ we have 
\al{ 
\norm{A-\widehat{A}_{\geq \sigma}}_{F} &\leq \norm{A-A_k}_F + (3\sqrt{ \eta} k^{1/4} \mu^{-1/4}+\sqrt{\mu/p}) \norm{A}_F.
}  
If $\norm{A-A_k}_F \leq \epsilon \norm{A}_F$ for some $\epsilon >0$ and $\norm{A}_{F} \geq \frac{ 36\sqrt{2} (nk)^{1/2}} { \epsilon^{3}}$ then we can choose $\eta, \mu$ such that $ \norm{A-\widehat{A}_{\geq \sigma}}_{F} \leq 3\epsilon \norm{A}_F$.
\end{theorem}

\begin{proof}
Let $\sigma_i$ denote the singular values of $\widehat{A}$. Let $\ell$ the largest integer for which $\sigma_\ell \geq \sqrt{ \frac{\mu}{k}}||\widehat{A}||_F$. Note that $\ell \leq \frac{k}{\mu}$. Then, by theorem \ref{AM}, we have
\begin{eqnarray*}
\norm{A-\widehat{A}_{\geq \sigma}}_{F}  =  \norm{A-\widehat{A}_{\ell}}_{F} 
 \leq  \norm{A- A_{\ell}}_{F}+  3\sqrt{ \eta} \ell^{1/4} \norm{A}_F .
\end{eqnarray*}

Define the random variable $X= \sum_{i,j} \widehat{A}_{ij}^{2}$ so that $X= \norm{\widehat{A}}_{F}^{2}$ and 
$E[X ] = \norm{A}_{F}^{2}/p$. The random variables $\widehat{A}_{ij}$ are independent, using the Chernoff bounds 
we have $\Pr [ \norm{\widehat{A}}_{F}^{2} > (1+ \beta) \norm{A}_{F}^{2}/p ] \leq e^{-\beta^{2} \norm{A}_{F}^{2}/3p}$
for $\beta \in [0, 1]$. The probability that $\norm{\widehat{A}}_{F}^{2} > 2\norm{A}_{F}^{2}/p$ is exponentially small.

We distinguish two cases.

If $\ell \geq k$, then $ \norm{A - A_\ell}_F \leq \norm{A - A_k}_F,$ since $A_\ell$ contains more of the singular vectors of $A$.

If $k > \ell$, then $ \norm{A - A_\ell}_F \leq \norm{A - A_k}_F + \norm{A_k-A_\ell}_F$, which dominates the two cases. For the second term we have $ \norm{A_k-A_\ell}_F^2 = \sum_{i=\ell+1}^k \sigma_i^2 \leq k \frac{\mu}{k}\norm{\widehat{A}}_F^2 \leq \frac{2\mu}{p} \norm{A}_F^2$.
Hence,
\[
\norm{A-\widehat{A}_{\geq \sigma}}_{F}  \leq \norm{A-A_k}_F + (3\sqrt{ \eta} k^{1/4} \mu^{-1/4}+\sqrt{2\mu/p}) \norm{A}_F.
\]
\noindent

If $\norm{A-A_k}_F \leq \epsilon \norm{A}_F$, for some $\epsilon \geq 0$ then we choose $\mu = \epsilon^{2} p/2$ and we can select any $\eta \leq \frac{2n^{1/4}  \epsilon^{3/2}} { 3 (2k)^{1/4} \norm{A}_{F}^{1/2}}$ so that 
$3\sqrt{ \eta} k^{1/4} \mu^{-1/4} \leq \epsilon$ and the overall error $ \norm{A-\widehat{A}_{\geq \sigma}}_{F} \leq 3\epsilon \norm{A}_F$. Indeed,
\als{ 
3\sqrt{ \eta} k^{1/4} \mu^{-1/4} = \frac{ 3\eta^{1/2} (2k)^{1/4} }{\epsilon^{1/2}  p^{1/4}  } = \frac{ 3\eta  \norm{A}_{F}^{1/2} (2k)^{1/4}  }{2\epsilon^{1/2} n^{1/4}  } \leq \epsilon  
} 
Note that for this choice of $\mu$ and $\eta$, the sampling probability must be at least $p \geq \frac{ 36\sqrt{2} (nk)^{1/2}} { \norm{A}_{F} \epsilon^{3}}$, the assumption in the theorem statement ensures that $p\leq 1$. 
\end{proof}

Our quantum procedure will almost produce this projection. In fact, we will need to consider a family of matrices which denote the projection of the matrix $A$ onto the space spanned by the union of the singular vectors whose corresponding singular value is bigger than $\sigma$ and also some subset of singular vectors whose corresponding singular value is in the interval $[(1-\kappa)\sigma, \sigma)$. Think of $\kappa$ as a constant, for example $1/3$. This subset could be empty, all such singular vectors, or any in-between subset. We denote by$A_{\geq \sigma, \kappa}$ any matrix in this family.

The final theorem we will need is the following

\begin{theorem}\label{sigmakappa}
Let $A\in \R^{m\times n}$ be a matrix and $\max_{ij} A_{ij}=1$. 
Define the matrix $\widehat{A}$ to be a random matrix obtained by subsampling with probability 
$p = 16 n / (\eta \norm{A}_F)^2$ and rescaling, that is $\widehat{A}_{ij} = A_{ij}/p$ with probability $p$ and $0$ otherwise. Let $\mu>0$ a threshold parameter and denote $\sigma=\sqrt{ \frac{\mu}{k}}||\widehat{A}||_F$. Let $\kappa >0$ a precision parameter. 
With probability at least 
$1- exp(-19 (\log n)^{4})$,
\al{ 
\norm{A-\widehat{A}_{\geq \sigma,\kappa}}_{F}  \leq 3\norm{A-A_k}_F + 
\left( 3\sqrt{ \eta} k^{1/4} \mu^{-1/4}(2+(1-\kappa)^{-1/2})+(3-\kappa)\sqrt{2\mu/p}\right) \norm{A}_F.
}  
If $\norm{A-A_k}_F \leq \epsilon \norm{A}_F$ for some $\epsilon >0$ and $\norm{A}_{F} \geq \frac{ 36\sqrt{2} (nk)^{1/2}} { \epsilon^{3}}$ then we can choose $\eta, \mu$ such that $ \norm{A-\widehat{A}_{\geq \sigma, \kappa}}_{F} \leq 9\epsilon \norm{A}_F$.
\end{theorem}

\begin{proof}
We have
\begin{eqnarray*}
\norm{A-\widehat{A}_{\geq \sigma,\kappa}}_{F} & \leq & \norm{A-\widehat{A}_{\geq \sigma}}_{F} + \norm{\widehat{A}_{\geq \sigma}-\widehat{A}_{\geq \sigma,\kappa}}_{F} \\
& \leq & 
\norm{A-\widehat{A}_{\geq \sigma}}_{F} + \norm{\widehat{A}_{\geq \sigma}-\widehat{A}_{\geq (1-\kappa)\sigma}}_{F}\\
& \leq & \norm{A-\widehat{A}_{\geq \sigma}}_{F} + \norm{A- \widehat{A}_{\geq \sigma}}_F+ \norm{A- \widehat{A}_{\geq (1-\kappa)\sigma}}_{F}\\
& \leq & 2 \norm{A-\widehat{A}_{\geq \sigma}}_{F} + \norm{A- \widehat{A}_{\geq (1-\kappa)\sigma}}_{F}.\\
\end{eqnarray*}
We use Theorem \ref{sigma} to bound the first term as
\[
\norm{A-\widehat{A}_{\geq \sigma}}_{F}  \leq \norm{A-A_k}_F + (3\sqrt{ \eta} k^{1/4} \mu^{-1/4}+\sqrt{2\mu/p}) \norm{A}_F
\]
For the second term, we can reapply Theorem \ref{sigma} where now we need to rename $\mu$ as $(1-\kappa)^2\mu$ and have
\[
\norm{A-\widehat{A}_{\geq (1-\kappa)\sigma}}_{F}  \leq \norm{A-A_k}_F + (3\sqrt{ \eta} k^{1/4} (1-\kappa)^{-1/2}\mu^{-1/4}+(1-\kappa)\sqrt{2\mu/p}) \norm{A}_F.
\]
Overall we have
\[
\norm{A-\widehat{A}_{\geq \sigma,\kappa}}_{F}  \leq 3\norm{A-A_k}_F + 
\left( 3\sqrt{ \eta} k^{1/4} \mu^{-1/4}(2+(1-\kappa)^{-1/2})+(3-\kappa)\sqrt{2\mu/p}\right) \norm{A}_F.
\]
Let $\norm{A-A_k}_F \leq \epsilon \norm{A}_F$, for some $\epsilon \geq 0$.
We choose $\kappa=1/3$, $\mu = \epsilon^{2} p/2$ and we can select any $\eta \leq \frac{2n^{1/4}  \epsilon^{3/2}} { 3 (2k)^{1/4} \norm{A}_{F}^{1/2}}$ to have
\begin{equation}\label{9epsilon}
\norm{A-\widehat{A}_{\geq \sigma,\kappa}}_{F} \leq 3\epsilon \norm{A}_F + \left(2\epsilon+\frac{\epsilon}{\sqrt{1-\kappa}}+(3-\kappa)\epsilon\right) \norm{A}_F \leq 9 \epsilon \norm{A}_F.
\end{equation}
As in theorem \ref{sigma}, the sampling probability must be at least $p \geq \frac{ 36\sqrt{2} (nk)^{1/2}} { \norm{A}_{F} \epsilon^{3}}$.
\end{proof}

We have shown that the task of providing good recommendations for a user $i$ reduces to being able to sample from the $i$-th row of the matrix $\widehat{T}_{\geq \sigma,\kappa}$, in other words sample from the projection of the $i$-th row of $\widehat{T}$ onto the space spanned by all row singular vectors with singular values higher than $\sigma$ and possibly some more row singular vectors with singular values in the interval $[(1-\kappa)\sigma, \sigma)$. 

In the following section, we show a quantum procedure, such that given a vector (e.g. the $i$-th row of $\widehat{T}$), a matrix (e.g. the matrix $\widehat{T}$), and parameters $\sigma$ and $\kappa$, outputs the quantum state $\ket{(\widehat{T}_{\geq \sigma,\kappa})_i}$, which allows one to sample from this row by measuring in the computational basis. The algorithm runs in time polylogarihmic in the matrix dimensions and polynomial in $k$, since it depends inverse polynomially in $\sigma$, which in our case is inverse polynomial in $k$. 


\section{Quantum projections in polylogarithmic time}

The main quantum primitive required for the recommendation system is a quantum projection algorithm that 
runs in time polylogarithmic in the matrix dimensions. 

\subsection{The data structure}
The input to the quantum procedure is a vector $x \in \R^n$ and a matrix $A \in \R^{m\times n}$. 
We assume that the input is stored in a classical data structure such that an algorithm that has quantum access to the data structure can create the quantum state $\ket{x}$ corresponding to the vector $x$ and the quantum states $\ket{A_i}$ corresponding to each row $A_i$ of the matrix $A$, in time $\mbox{polylog}(mn)$. 

It is in fact possible to design a data structure for a matrix $A$ that supports the efficient construction of the quantum states $\ket{A_i}$. Moreover, we can ensure that the size of the data structure is optimal (up to polylogarithmic factors), and the data entry time, i.e. the time to store a new entry $(i, j, A_{ij})$ that arrives in the system is just $\mbox{polylog}(mn)$. Note that just writing down the entry takes logarithmic time. 

\begin{theorem}\label{datastr}
Let $A\in \R^{m\times n}$ be a matrix. Entries $(i, j, A_{ij})$ arrive in the system in an arbitrary order and $w$ denotes the number of entries that have already arrived in the system. There exists a data structure to store the entries of $A$ with the following properties:
\begin{enumerate}[i.]
\item The size of the data structure is $O(w\cdot\log^2(mn))$.
\item The time to store a new entry $(i,j,A_{ij})$ is $O(\log^2(mn))$.
\item A quantum algorithm that has quantum access to the data structure can 
perform the mapping $\widetilde{U}:\ket{i}\ket{0} \to \ket{i} \ket{A_i}$, for $i \in [m]$,
corresponding to the rows of the matrix currently stored in memory and the mapping $\widetilde{V}:\ket{0}\ket{j} \to \ket{\widetilde{A}} \ket{j}$, for $j \in [n]$,
where $\widetilde{A} \in \R^{m}$ has entries $\widetilde{A}_{i} = \norm{A_{i}} $ in time $\emph{polylog}(mn)$. 
\end{enumerate}
\end{theorem}

The explicit description of the data structure is given in the appendix. Basically, for each row of the matrix, that we view as a vector in $\R^n$, we store an array of $2n$ values as a full binary tree of $n$ leaves. The leaves hold the individual amplitudes of the vector and each internal node holds the sum of the squares of the amplitudes of the leaves rooted on this node. For each entry added to the tree, we need to update $\log(n)$ nodes in the tree. The same data structure can of course be used for the vector $x$ as well. One need not use a fixed array of size $2n$ for this construction, but only ordered lists of size equal to the entries that have already arrived in the system.Alternative solutions for vector state preparation are possible, another solution based on a modified memory is described in \cite{P14}.

\subsection{Quantum Singular Value Estimation}
The second tool required for the projection algorithm is an efficient quantum algorithm for singular value estimation. In the singular value estimation problem we are given a matrix $A$ such that the vector states corresponding to its row vectors can be prepared efficiently. Given a state $\ket{x}=\sum_{i} \alpha_{i} \ket{v_{i}} $ for an arbitrary vector $x\in \R^{n}$ the task is to estimate the singular values corresponding to each singular vector in coherent superposition. Note that we take the basis $\{v_i\}$ to span the entire space by including singular vectors with singular value 0. 


\begin{theorem} \label{tsve} 
Let $A \in \R^{m \times n}$ be a matrix with singular value decomposition $A= \sum_{i} \sigma_{i} u_{i} v_{i}^{t}$ stored in the data structure in theorem \ref{datastr}. Let $\epsilon>0$ be the precision parameter. There is an algorithm with running time $O(\emph{polylog}(mn)/\epsilon)$ that performs the mapping  $\sum_{i} \alpha_{i} \ket{v_{i}}
\to \sum_{i} \alpha_{i} \ket{v_{i}} \ket{ \overline{\sigma_{i}}}$, where $\overline{\sigma_{i}} \in \sigma_{i} \pm \epsilon  \norm{A}_{F}$ for all $i$ with probability at least $1- 1/\emph{poly}(n)$.  
\end{theorem} 
Here, we present a quantum singular value estimation algorithm, in the same flavor as the quantum walk based algorithm by Childs \cite{C10} for estimating eigenvalues of a matrix, and show that given quantum access to the  data structure from theorem \ref{datastr}, our algorithm runs in time  $O(\text{polylog} (mn)/\epsilon)$. 
A different quantum algorithm for singular value estimation can be based on the work of \cite{LMR13} with running time $O(\text{polylog} (mn)/\epsilon^{3})$, and for which a coherence analysis was shown in \cite{P14}. 

The idea for our singular value estimation algorithm is to find isometries $P \in \R^{mn \times m}$ and $Q \in \R^{mn \times n}$ that can be efficiently applied, and such that $\frac{A}{\norm{A}_{F}}=P^{t}Q$. Using $P$ and $Q$, we define a unitary matrix $W$ acting on $\R^{mn}$, which is also efficiently implementable and such that the row singular vector $v_i$ of $A$ with singular value $\sigma_i$ is mapped to an eigenvector $Qv_i$ of $W$ with eigenvalue $e^{\iota \theta_{i}}$ such that $\cos(\theta_{i}/2)= \sigma_{i}/\norm{A}_{F}$ (note that $\cos(\theta_{i}/2)>0$ as $\theta_{i} \in [-\pi, \pi]$). 
The algorithm consists of the following steps: first, map the input vector $\sum_{i} \alpha_{i} \ket{v_{i}}
$ to $\sum_{i} \alpha_{i} \ket{Qv_{i}}$ by applying $Q$; then, use phase estimation as in theorem \ref{pest} with 
unitary $W$ to compute an estimate of the eigenvalues $\theta_{i}$ and hence of the singular values $\sigma_{i} = \norm{A}_F \cos(\theta_{i}/2)$; and finally undo $Q$ to recover the state $\sum_{i} \alpha_{i} \ket{v_{i}}\ket{\sigma_i}$. This procedure is described in algorithm \ref{algqsve}. 

It remains to show how to construct the mappings $P,Q$ and the unitary $W$ that satisfy all the properties mentioned above that are required for the quantum singular value estimation algorithm. 

\begin{lemma} \label{l53} 
Let $A \in \R^{m \times n}$ be a matrix with singular value decomposition $A= \sum_{i} \sigma_{i} u_{i} v_{i}^{t}$ stored in the data structure in theorem \ref{datastr}. Then, there exist matrices $P \in \R^{mn \times m}, Q \in \R^{mn \times n}$ such that
\begin{enumerate}[i.]
\item The matrices $P,Q$ are a factorization of $A$, i.e. $\frac{A}{\norm{A}_{F}} = P^{t} Q$. Moreover, $P^{t}P=I_{m}$, $Q^{t}Q=I_{n}$, and multiplication by $P,Q$, i.e. the mappings $\ket{y} \to \ket{Py}$ and $\ket{x} \to \ket{Qx}$ can be performed in time $O(\emph{polylog} (mn))$.
\item The unitary $W=U \cdot V$, where $U,V$ are the reflections $U= 2PP^{t} - I_{mn}$ and $V= 2QQ^{t} - I_{mn}$ can be implemented in time $O(\emph{polylog} (mn))$.
\item The isometry $Q : \R^n \to \R^{mn}$ maps a row singular vector $v_i$ of $A$ with singular value $\sigma_i$ to an eigenvector $Qv_i$ of $W$ with eigenvalue $e^{\iota \theta_{i}}$ such that $\cos(\theta_{i}/2)= \sigma_{i}/\norm{A}_{F}$. 
\end{enumerate}
\end{lemma}

\begin{proof} 
Let $P \in \R^{mn \times m}$ be a matrix with column vectors $e_{i}  \otimes \frac{A_{i}}{\norm{A_{i}}}$ for $i \in [m]$. In quantum notation multiplication by $P$ can be expressed as 
\[
\ket{Pe_{i}} = \ket{i, A_{i}} = \frac{1}{\norm{A_i}} \sum_{j\in [n]} A_{ij} \ket{i,j}, \quad \mbox{ for } i \in [m]. 
\]
Let $\widetilde{A} \in \R^{m}$ be the vector of Frobenius norms of the rows of the matrix $A$, that is $\widetilde{A}_{i} = \norm{A_{i}}$ for $i \in [m]$. Let $Q \in \R^{mn \times n}$ be a matrix with column vectors $\frac{\widetilde{A}}{\norm{A}_{F}}  \otimes e_{j}$ for $j \in [n]$. In quantum notation multiplication by $Q$ can be expressed as 
\[
\ket{Qe_{j}} = \ket{\widetilde{A}, j} = \frac{1}{\norm{A}_F}  \sum_{i \in [m]}  \norm{A_i} \ket{i,j}, \quad \mbox{ for } j \in [n]. 
\]
The factorization $A=P^{t}Q$ follows easily by expressing the matrix product in quantum notation, 
\[ 
(P^{t} Q)_{ij} = \braket{ i, A_{i}} { \widetilde{A}, j } = \frac{ \norm{A_{i}} } {\norm{A}_{F}} \frac{ A_{ij} } { \norm{A_{i}} }   = \frac{A_{ij} } {\norm{A}_{F} }.
\]
The columns of $P, Q$ are orthonormal by definition so $P^{t} P= I_{m}$ and $Q^{t} Q =I_{n}$. Multiplication by $P$
and $Q$  can be implemented in time $\text{polylog}(mn)$ using quantum access to the data structure from theorem \ref{datastr}, 
\all{ 
\ket{y} &\to \ket{y, 0^{\lceil \log n \rceil } } = \sum_{i \in [m]} y_{i}\ket{i, 0^{\lceil \log n \rceil }  } \xrightarrow{\widetilde{U}}  \sum_{i \in [m]} y_{i}\ket{i, A_{i} } = \ket{Py} \nl 
\ket{x} &\to \ket{0^{\lceil \log m \rceil }, x } = \sum_{j \in [n]} x_{j}\ket{0^{\lceil \log m \rceil }, j }  \xrightarrow{\widetilde{V}}  \sum_{j \in [n]} x_{j}\ket{\widetilde{A}, j } = \ket{Qx}.
} {ds1} 
To show (ii), note that the unitary $U$ is a reflection in $Col(P)$ and can be implemented as $U=\widetilde{U}R_{1}\widetilde{U}^{-1}$ where 
$\widetilde{U}$ is the unitary in first line of equation \eqref{ds1} and $R_{1}$ is the reflection in the space $\ket{y, 0^{\lceil \log n \rceil } }$ for $y \in \R^m$.  
It can be implemented as a reflection conditioned on the second register being in state $\ket{0^{\lceil \log n \rceil }}$. The unitary 
$V$ is a reflection in $Col(Q)$ and can be implemented analogously as $V=\widetilde{V}R_{0}\widetilde{V}^{-1}$ where 
$\widehat{V}$ is the unitary in the second line of equation \eqref{ds1} and $R_{0}$ is the reflection in the space $\ket{0^{\lceil \log m \rceil }, x }$ 
for $x \in \R^{n}$.

It remains to show that $Qv_{i}$ is an eigenvector for $W$ with eigenvalue $e^{\iota \theta_{i}}$ such that 
$\cos(\theta_{i}/2) = \sigma_{i}/\norm{A}_{F}$. For every pair of singular vectors $(u_{i}, v_{i})$ of $A$, we define the two dimensional subspaces $\mathcal{W}_{i}=Span(Pu_{i}, Qv_{i})$ 
and let $\theta_{i}/2 \in [-\pi/2, \pi/2]$ be the angle between $Pu_{i}$ 
and $\pm Qv_{i}$. Note that $\mathcal{W}_{i}$ is an eigenspace for $W$ which acts on it as a rotation by $\pm \theta_{i}$,  since $W$ is a reflection in the column space of $Q$ followed by a reflection in the column space of $P$. Moreover, the relation $\cos(\theta_{i}/2) = \sigma_{i}/\norm{A}_{F}$ 
is a consequence of the factorization in lemma \ref{l53}, since we have
\al{ 
PP^{t} Q v_{i} = \frac{PA v_{i}}{ \norm{A}_{F}} = \frac{\sigma_{i}}{\norm{A}_{F} }  Pu_{i} \quad \mbox{and} \quad
QQ^{t} P u_{i} = \frac{QA^{t} u_{i}}{ \norm{A}_{F}} = \frac{\sigma_{i}}{\norm{A}_{F} }  Qv_{i}.
}

\end{proof} 
Using the primitives from the preceding lemma, we next describe the singular value estimation algorithm and analyze it to prove 
theorem \ref{tsve}.

\begin{algorithm}[H]
\caption{Quantum singular value estimation} \label{algqsve}
\begin{algorithmic}[1]
\REQUIRE $A \in \R^{m\times n}$, $x \in \R^{n}$ in the data structure from theorem \ref{datastr}, precision parameter $\epsilon>0$.  \\
\begin{enumerate} 
\item Create $\ket{x} = \sum_{i} \alpha_{i} \ket{v_{i}}$. \\
\item Append a first register $\ket{0^{\lceil \log m \rceil }}$ and create the state $\ket{Qx} = \sum_{i} \alpha_{i} \ket{Qv_{i}}$ as in eq. \eqref{ds1}. \\
\item Perform phase estimation 
with precision parameter $2\epsilon >0$
on the input $\ket{Qx}$ for the unitary $W=U\cdot V$ where 
$U,V$ are the unitaries in lemma \ref{l53} and obtain $\sum_{i} \alpha_{i} \ket{Qv_{i}, \overline{\theta_{i}} } $.  
\item Compute $\overline{\sigma_{i}} = \cos(\overline{\theta_{i}}/2) \norm{A}_{F}$ where 
 $\overline{\theta_{i}}$ is the estimate from phase estimation, and uncompute the 
output of the phase estimation. \nl 
\item Apply the inverse of the transformation in step $2$ to obtain $\sum_{i} \alpha_{i} \ket{v_{i} } \ket{\overline{\sigma_{i}}}$. \nl 
 \end{enumerate} 
\end{algorithmic}
\end{algorithm}

\paragraph{Analysis}
The phase estimation procedure \cite{K95} with unitary $W$ and precision parameter $\epsilon$ on input $\ket{Qv_{i}}$ produces an estimate 
such that $|\overline{\theta_{i}} - \theta_{i} | \leq 2\epsilon$. The estimate for the singular value is $\overline{\sigma_{i}}= \cos(\overline{\theta_{i}}/2)\norm{A}_{F}$. The error 
in estimating $\sigma_{i}= \cos(\theta_{i}/2)\norm{A}_{F}$ can be bounded as follows, 
\al{ 
|\overline{\sigma_{i} } - \sigma_{i} | =  | \cos(\theta_{i}/2) - \cos(\overline{\theta_{i}}/2) | \norm{A}_{F}\leq \sin(\phi) \frac{|\overline{\theta_{i}} - \theta_{i}|}{2}   \norm{A}_{F}  \leq    \epsilon \norm{A}_{F}
} 
where $\phi \in [\theta_{i}/2 - \epsilon, \theta_{i}/2 + \epsilon]$. 
Algorithm \ref{algqsve} therefore produces an additive error $\epsilon \norm{A}_{F}$ estimate of the singular values, the running time is $O(\text{polylog}(mn)/\epsilon)$ by theorem \ref{pest} as the unitary $W$ is implemented in time $O(\text{polylog}(mn))$ by lemma 
\ref{l53}. This concludes the proof of theorem \ref{tsve}.

One can define an algorithm for singular value estimation with input $\ket{y}=\sum_{i} \beta_{i} \ket{u_{i}}$. where $u_i$ are the column singular vectors, by using the operator $P$ from lemma \ref{l53} instead of $Q$ in algorithm \ref{algqsve}. The correctness follows from the same argument as above.


\subsection{Quantum projection with threshold} \label{qp} 
Let $A = \sum_{i} \sigma_i u_i v_i^t$. We recall that $A_{\geq \sigma} = \sum_{\sigma_{i} \geq \sigma} \sigma_i u_i v_i^t$ is the projection of the matrix $A$ onto the space spanned by the singular vectors whose singular values are bigger than $\sigma$. Also, $A_{\geq \sigma,\kappa}$ is the projection of the matrix $A$ onto the space spanned by the union of the singular vectors whose corresponding singular values is bigger than $\sigma$ and some subset of singular vectors whose corresponding singular values are in the interval $[(1-\kappa)\sigma,\sigma)$.
 
Algorithm \ref{projection} presents a quantum algorithm that given access to vector state $x$, a matrix $A$ and parameters $\sigma,\kappa$, outputs the state $\ket{A^+_{\geq \sigma,\kappa} A_{\geq \sigma,\kappa}  x}$, namely
the projection of $x$ onto the subspace spanned by the union of the row singular vectors whose corresponding singular values are bigger than $\sigma$ and some subset of row singular vectors whose corresponding singular values are in the interval $[(1-\kappa)\sigma,\sigma)$. 

For simplicity, we present the algorithm without a stopping condition and we will compute the expected running time. By stopping the algorithm after a number of iterations which is $\log(n)$ times more than the expected one, we can easily construct an algorithm with worst-case running time guarantees and whose correctness probability has only decreased by a factor of $(1-1/\text{poly}(n))$.

Let $\{v_i\}$ denote an orthonormal basis for $\R^n$ that includes all row singular vectors of the matrix $A$. 
We think of $\kappa$ as a constant, for example $1/3$. 

\begin{algorithm}[H]\label{projection}
\caption{Quantum projection with threshold} 
\begin{algorithmic}[1]
\REQUIRE $A \in \R^{m\times n}$, $x \in \R^{n}$ in the data structure from Theorem \ref{datastr}; parameters $\sigma,\kappa >0$. \vspace{0.2cm}

\begin{enumerate}
\item Create $\ket{x} = \sum_i \alpha_i \ket{v_i}$. \nl 
\item Apply the singular value estimation on $\ket{x}$ with precision $\epsilon \hat{=} \frac{\kappa}{2} \frac{ \sigma}{\norm{A}_F}$ to obtain the state
\[
\sum_i \alpha_i \ket{v_i}\ket{\overline{\sigma}_i}
\]
\item Apply on a second new register the unitary $V$ that maps $\ket{t}\ket{0} \mapsto \ket{t}\ket{1}$ if $t < \sigma - \frac{\kappa}{2} \sigma$ and $\ket{t}\ket{0} \mapsto \ket{t}\ket{0}$ otherwise, to get the state
\[
\sum_{i \in S} \alpha_i \ket{v_i}\ket{\overline{\sigma}_i}\ket{0} + \sum_{i \in \overline{S}} \alpha_i \ket{v_i}\ket{\overline{\sigma}_i}\ket{1},
\]
where $S$ is the union of all $i$'s such that $\sigma_i \geq \sigma$ and some $i$'s with $\sigma_i \in [(1-\kappa)\sigma,\sigma)$.
\item Apply the singular value estimation on the above state to erase the second register 
\[
\sum_{i \in S} \alpha_i \ket{v_i}\ket{0} + \sum_{i \in \overline{S}} \alpha_i \ket{v_i}\ket{1} = 
\beta
\ket{A^+_{\geq \sigma,\kappa} A_{\geq \sigma,\kappa} x} \ket{0} + 
\sqrt{1-|\beta|^2}
\ket{A^+_{\geq \sigma,\kappa} A_{\geq \sigma,\kappa} x}^\bot \ket{1},
\]
with $ \beta = \frac{|| A^+_{\geq \sigma,\kappa} A_{\geq \sigma,\kappa} x||}{\norm{x} }$.
\item Measure the second register in the standard basis. If the outcome is $\ket{0}$, output the first register and exit. Otherwise repeat step 1.
\end{enumerate}
\end{algorithmic}
\end{algorithm}

For the running time, note that the singular value estimation takes time $O(\text{polylog}(mn)/\epsilon)$, while the probability we obtain $\ket{0}$ in step 5 is 
$\frac{||A^+_{\geq\sigma,\kappa} A_{\geq\sigma,\kappa} x||^{2} }{\norm{x} ^2} \geq \frac{||A^+_{\geq\sigma} A_{\geq\sigma} x||^{2}}{\norm{x} ^2}$.

\begin{theorem} \label{aqproj} 
Algorithm \ref{projection} outputs $\ket{A^+_{\geq\sigma,\kappa} A_{\geq\sigma,\kappa} x}$ with probability at least $1 - 1/\emph{poly}(n)$ and in expected time $O(\frac{ \emph{polylog}(mn) \norm{A}_F \norm{x} ^2}{\sigma ||A_{\geq\sigma} A_{\geq\sigma} ^{+} x||^2})$. 
\end{theorem} 

It is important to notice that the running time of the quantum projection algorithm depends only on the threshold $\sigma$ (which we will take to be of the order $\frac{\norm{A}_F}{\sqrt{k}}$) and not on the condition number of $A$ which may be very large. We will also show in the next section that in the recommendation systems, for most users the ratio $\frac{||A^+_{\geq\sigma} A_{\geq\sigma} x||^{2}}{\norm{x} ^2}$ is constant.  This will conclude the analysis and show that the running time of the quantum recommendation system is polynomial in $k$ and polylogarithmic in the matrix dimensions. 

One could also use amplitude amplification to improve the running time of algorithm \ref{projection}, once a careful error analysis is performed as the reflections are not exact. As we will see that the $\frac{||A^+_{\geq\sigma} A_{\geq\sigma} x||^{2}}{\norm{x} ^2}$ is constant for most users, this will not change asymptotically the running time of the algorithm and hence we omit the analysis.


\section{Quantum recommendation systems}

We have all the necessary ingredients to describe the quantum algorithm that provides good recommendations for a user $i$ and that runs in time polylogarithmic in the dimensions of the preference matrix and polynomial in the rank $k$. As we said, in recommendation systems we assume that for the matrix $T$ we have $\norm{T - {T}_k}_F \leq \epsilon ||T||_F$ for some small approximation parameter $\epsilon$ and small rank $k$ (no more than 100). In the algorithm below, as in the classical recommendation systems, we assume we know $k$ but in fact we just need to have a good estimate for it. 

Note again that we do not put a stopping condition to the algorithm and we compute the expected running time. Again we can turn this into an algorithm with worst-case running time guarantees by stopping after running for $\log(n)$ times more than the expected running time, and the correctness probability has only decreased by a factor of $1-1/\text{poly}(n)$.

\begin{algorithm}[H]
\caption{Quantum recommendation algorithm.} \label{qproj}
\label{lpp}
\begin{algorithmic}[1]
\REQUIRE A subsample matrix $\widehat{T}\in \R^{m\times n}$ (with sampling probability $p$) stored in the data structure from Theorem \ref{datastr} and satisfying the conditions in Theorem \ref{sigmakappa}; a user index $i$. 
\STATE Apply the quantum projection procedure \ref{projection} with the matrix $\widehat{T}$, the vector corresponding to the $i$-th row $\widehat{T}_i$, with $\sigma=\sqrt{\frac{\epsilon^2 p}{2k}}\norm{\widehat{T}}_F $ and 
$\kappa=1/3$. \\
 The algorithm runs in expected time $O(\text{polylog}(mn)\sqrt{k} \norm{\widehat{T}_i}^2 / \sqrt{p} \norm{\widehat{T}^+_{\geq\sigma} \widehat{T}_{\geq\sigma} \widehat{T}_i }^2)$ and returns with probability at least $1-1/\text{poly}(n)$ the state 
$\ket{\widehat{T}^+_{\geq \sigma,\kappa} \widehat{T}_{\geq \sigma,\kappa} \widehat{T}_i}$.
\STATE  Measure the above state in the computational basis to get a product $j$. 
\end{algorithmic}
\end{algorithm}

\subsection{Analysis}

\paragraph{Correctness}
Let us check the correctness of the algorithm. 
Note that 
$\widehat{T}^+_{\geq \sigma,\kappa} \widehat{T}_{\geq \sigma,\kappa} \widehat{T}_i = (\widehat{T}_{\geq \sigma,\kappa})_i$, i.e. the $i$-th row of the matrix $\widehat{T}_{\geq \sigma,\kappa}$. Hence, the quantum projection procedure outputs with probability at least $1-1/\text{poly}(n)$ the state $\ket{(\widehat{T}_{\geq \sigma,\kappa})_i}$, meaning that our quantum recommendation algorithm with high probability outputs a product by sampling the $i$-th row of the matrix $\widehat{T}_{\geq \sigma,\kappa}$.

By Theorem \ref{sigmakappa}, and by setting the parameters appropriately to get equation \eqref{9epsilon}, we have that with probability at least $1- exp(-19 (\log n)^{4})$,
\[
\norm{T-\widehat{T}_{\geq \sigma,\kappa}}_F \leq 9\epsilon \norm{T}_F.
\] 
In this case, we can apply Theorem \ref{rec} with matrix $\widetilde{T}=\widehat{T}_{\geq \sigma,\kappa}$ to show that there exists a subset of users $S' $ of size at least $(1-\delta-\zeta)m$ (for $\delta>0$), such that on average over the users in $S'$, the probability that our quantum algorithm  provides a bad recommendation is 
\[ \Pr_{i \sim \mathcal{U}_{S'}, j \sim (\widehat{T}_{\geq \sigma,\kappa})_i}[(i,j) \text{ bad}] \leq \frac{\left( \frac{9\epsilon(1+9\epsilon)}{1-9\epsilon} \right)^2}{ \left( 1/\sqrt{1+\gamma} - 9\epsilon/\sqrt{\delta}\right)^2  (1-\delta - \zeta)}.
\]

\paragraph{Expected running time}
We prove the following theorem
\begin{theorem}
For at least $(1-\xi)(1-\delta-\zeta)m$ users in the subset $S'$, we have that the expected running time of Algorithm \ref{lpp} is $O(\emph{polylog}(mn)\emph{poly}(k))$.
\end{theorem}
\begin{proof}

First, by the conditions of Theorem \ref{sigmakappa}, we must have $p \geq \frac{ 36\sqrt{2} (nk)^{1/2}} { \norm{A}_{F} \epsilon^{3}}$. That is the theorem works even for a $p$ which is sub-constant. However, in order to have the desired running time, we need to take $p$ to be some constant, meaning that we need to subsample a constant fraction of the matrix elements. This is also the case for classical recommendation systems \cite{AFKMS01, DKR02}. 

Second, we need to show that for most users the term $W_i \equiv \frac{||\widehat{T}_i||^2}{ ||(\widehat{T}_{\geq \sigma,\kappa})_i||^2}$ that appears in the running time of the quantum projection algorithm is a constant. This is to be expected, since most typical rows of the matrix project very well onto the space spanned by the top singular vectors, since the spectrum of the matrix is well concentrated on the space of the top singular vectors.
%

As in Theorem \ref{rec}, we focus on the users in the subset $S'$, with $|S'| \geq (1-\delta - \zeta)m$, for which equations \ref{assumption} and \ref{normS} hold. For these users we can use equation \ref{normS}
with the matrix $\widetilde{T} = \widehat{T}_{\geq \sigma,\kappa}$ and error $9\epsilon$. We have

\begin{eqnarray*}
E_{i \in S'}[W_i] & = & E_{i \in S'}[\frac{||\widehat{T}_i||^2} { ||(\widehat{T}_{\geq \sigma,\kappa})_i||^2}] 
\leq
\frac{E_{i \in S'}[||\widehat{T}_i||^2]}{\frac{\norm{\widehat{T}}_{F}^{2}}{(1+\epsilon)^2m} \en{ \frac{1}{\sqrt{1+\gamma}} - \frac{9\epsilon}{\sqrt{\delta}} }^{2}} 
\leq \frac{\frac{||\widehat{T}||_F^2}{(1-\delta-\zeta)m}}{\frac{\norm{\widehat{T}}_{F}^{2}}{(1+\epsilon)^2m} \en{ \frac{1}{\sqrt{1+\gamma}} -\frac{9\epsilon}{\sqrt{\delta}} }^{2}} \\
& \leq & \frac{(1+\epsilon)^2}{ (1-\delta-\zeta)\en{ \frac{1}{\sqrt{1+\gamma}} - \frac{9\epsilon}{\sqrt{\delta}} }^{2}}.
\end{eqnarray*}
By Markov's inequality, for at least $(1-\xi)|S'|$ users in $S'$ we have  $W_i \leq \frac{(1+\epsilon)^2}{\xi (1-\delta-\zeta)\en{\frac{1}{\sqrt{1+\gamma}} - \frac{9\epsilon}{\sqrt{\delta}} }^{2}}$, which for appropriate parameters is a constant. 
Hence, for at least $(1-\xi)(1-\delta - \zeta)m$ users, the quantum recommendation algorithm has an expected running time of $O(\text{poly}(k)\text{polylog}(mn))$ and produces good recommendations with high probability. As we said we can easily turn this into a worst-case running time, by stopping after running $\log (n)$ times more than the expected running time and hence decreasing the correctness only by a factor of $1-1/\text{poly}(n)$.
\end{proof}

\subsection*{Acknowledgements:} IK was partially supported by projects ANR RDAM, ERC QCC and EU QAlgo. 
AP was supported by the Singapore National Research Foundation under NRF RF Award No. NRF-NRFF2013-13.

\bibliographystyle{IEEEtranS} 
\bibliography{bibliography.bib}

\appendix

\section{The data structure} 
We prove the following theorem.
\begin{theorem}(Theorem \ref{datastr} restated)
Let $A\in \R^{m\times n}$ be a matrix. Entries $(i, j, A_{ij})$ arrive in the system in some arbitrary order, and $w$ denotes the number of entries that have already arrived in the system. There exists a data structure to store the matrix $A$ with the following properties:
\begin{enumerate}[i.]
\item The size of the data structure is $O(w \log^{2} (mn))$.
\item The time to store a new entry $(i,j,A_{ij})$ is $O(\log^2(mn))$.
\item A quantum algorithm that has quantum access to the data structure can 
perform the mapping $\widetilde{U}:\ket{i}\ket{0} \to \ket{i} \ket{A_i}$, for $i \in [m]$,
corresponding to the rows of the matrix currently stored in memory and the mapping $\widetilde{V}:\ket{0}\ket{j} \to \ket{\widetilde{A}} \ket{j}$, for $j \in [n]$,
where $\widetilde{A} \in \R^{m}$ has entries $\widetilde{A}_{i} = \norm{A_{i}} $ in time $\emph{polylog}(mn)$. 
\end{enumerate}
\end{theorem}
\begin{proof} 
The data structure consists of an array of $m$ binary trees $B_{i}, i \in [m]$. The trees $B_{i}$ are initially empty. When a new entry $(i,j,A_{ij})$ arrives the leaf node $j$ in tree $B_{i}$ 
is created if not present and updated otherwise. The leaf stores the value $A_{ij}^2$ as well as the sign of $A_{ij}$. 
The depth of each tree $B_{i}$ is at most $\lceil \log n \rceil$ as there can be at most $n$ leaves.  An internal node $v$ of $B_{i}$ stores the sum of the values of all leaves in the subtree rooted at $v$, i.e. the sum of the square amplitudes of the entries of $A_{i}$ in the subtree. Hence, the value stored at the root is $\norm{A_i}^2$. 
When a new entry arrives, all the nodes on the path from that leaf to the tree root are also updated. The different levels of the tree $B_{i}$ are stored as ordered lists 
so that the address of the nodes being updated can be retrieved in time $O(\log mn)$. 
The binary tree for a 4-dimensional unit vector for which all entries have arrived is illustrated in figure \ref{figtree}. 

The time required to store entry $(i,j,A_{ij})$ is $O(\log^{2} mn)$ as the insertion algorithm makes 
at most $\lceil \log n \rceil$ updates to the data structure and each update requires time 
$O(\log mn)$ to retrieve the address of the updated node.

\begin{figure} 

\begin{tikzpicture}[scale=1, level 1/.style={sibling distance=8em},
level 2/.style={sibling distance=5em}, level distance=2.5cm, information text/.style={rounded corners,
fill=green!5,inner sep=1ex}]] 

\node [rectangle, draw=black, minimum size=4mm]{1.0} 
child {node [rectangle, draw=black, minimum size=4mm] {0.32}  
	child { node [rectangle, draw=black]  {0.16} } 
	child { node [rectangle, draw=black]{0.16} }
	}
child {node [rectangle, draw=black, minimum size=4mm] {0.68} 
	child { node [rectangle, draw=black]{0.64} } 
	child { node [rectangle, draw=black]{0.04} } 	
        };
        
 \draw[xshift=4.2cm, yshift=-2cm] node[right,text width=8.5cm,information text] {
Let $\ket{\phi} = 0.4 \ket{00} + 0.4\ket{01} + 0.8\ket{10} + 0.2{\ket{11}}$.
\begin{itemize} 
\item Rotation on qubit 1: $ \ket{0}\ket{0} \to ( \sqrt{0.32} \ket{0} + \sqrt{0.68} \ket{1} ) \ket{0}$
\item Rotation on qubit 2 conditioned on qubit 1: 
\al{ 
 ( \sqrt{0.32} \ket{0} + \sqrt{0.68} \ket{1} ) \ket{0} \to \nl 
 \sqrt{0.32} \ket{0}  \frac{1}{\sqrt{0.32}}(0.4 \ket{0} +0.4 \ket{1} ) + \nl 
\sqrt{0.68} \ket{1} \frac{1}{\sqrt{0.68}} ( 0.8 \ket{0} + 0.2 \ket{1}) \notag 
}
\end{itemize}    
};
\end{tikzpicture}
\caption{Vector state preparation illustrated for $4$-dimensional state $\ket{\phi}$.} \label{figtree} 
\end{figure}

The memory requirement for the data structure is $O(w \log^2 mn)$ as for each entry $(i,j,A_{ij})$ at most 
 $\lceil \log n \rceil$ new nodes are added, each node requiring $O(\log mn)$ bits. 
 
 We now show how to perform $\widetilde{U}$ in time $\text{polylog}(mn)$ if an algorithm has quantum access to this classical data structure. The state preparation procedure using pre-computed amplitudes is well known in the literature, for instance see \cite{GR02}. 
The method is illustrated for a $4$-dimensional state $\ket{\phi}$ corresponding to a unit vector in figure \ref{figtree}. 
The amplitudes stored in the internal nodes of $B_{i}$ are used to apply a sequence of conditional rotations to the initial state $\ket{0}^{\lceil \log n \rceil}$ to obtain $\ket{A_{i}}$. Overall, there are $\lceil \log n \rceil$ rotations applied and for each one of them we need two quantum queries to the data structure (from each node in the superposition we query its two children).

The amplitude stored at an internal node of $B_{i}$ at depth $t$ corresponding to $k \in \{0, 1\}^{t}$ is,  
\[ 
B_{i,k} := \sum_{j \in [n], j_{1:t}=k}  A_{ij}^{2}
\]  
where $j_{1:t}$ denotes the first $t$ bits in the binary representation for $j$. Note that $B_{i,k}$ is the probability of observing outcome $k$ if the first $t$ bits of $\ket{A_{i}}$ are measured in the standard basis. 
Conditioned on the first register being $\ket{i}$ and the first $t$ qubits being in state $\ket{k}$ the rotation is applied to the $(t+1)$ qubit as follows
\[
\ket{i}\ket{k}\ket{0} \to \ket{i}\ket{k} \frac{1}{\sqrt{B_{i,k}}} \left( \sqrt{B_{i,k0}} \ket{0} + \sqrt{B_{i,k1}} \ket{1} \right).
\] 
The sign is included for rotations applied to the $\lceil \log n \rceil$-th qubit 
\[
\ket{i}\ket{k}\ket{0} \to \ket{i}\ket{k} \frac{1}{\sqrt{B_{i,k}}} \left( sgn(A_{k0})\sqrt{B_{i,k0}} \ket{0} + sgn(A_{k1})\sqrt{B_{i,k1}} \ket{1} \right). 
\]

Last, we show how to perform $\widetilde{V}$ in time $\text{polylog}(mn)$. Note that the amplitudes of the vector $\widetilde{A}$ are equal to $\norm{A_{i}}$, and the values stored on the roots of the trees $B_i$ are equal to $\norm{A_i}^2$. Hence, by a similar construction (another binary tree) for the $m$ roots, we can perform the unitary $\widetilde{V}$ efficiently.  

\end{proof} 

\end{document}